\documentclass[10pt]{IEEEtran}
\usepackage{epsfig,color,amsmath,cite}
\usepackage{amsthm}
\usepackage{amsmath}    
\IEEEoverridecommandlockouts
\usepackage{wrapfig}
\usepackage{bm}
\usepackage{epstopdf}
\usepackage{amssymb}
\usepackage{url}
\usepackage{enumitem}
\usepackage{multirow}
\usepackage{booktabs}
\usepackage{algorithm,algorithmic}	
\usepackage{comment}
\setlist[itemize]{leftmargin=*}

\DeclareMathOperator*{\minimize}{minimize}

\DeclareMathOperator{\subjectto}{subject\ to}

\makeatother
\DeclareMathAlphabet\mathbfcal{OMS}{cmsy}{b}{n}

\newtheorem{theorem}{Theorem}
\newtheorem{mydef}{Definition}

\newtheorem{asmp}{Assumption}

\newtheorem{myprs}{Proposition}
\newtheorem{exmpl}{Example}



\usepackage{stackengine}

\newcommand{\mat}[1]{\boldsymbol{#1}}

\newcommand{\bmat}[1]{\begin{bmatrix} #1 \end{bmatrix}}

\providecommand{\mA}{\ensuremath{\mat{A}}}
\providecommand{\mB}{\ensuremath{\mat{B}}}
\providecommand{\mC}{\ensuremath{\mat{C}}}

\providecommand{\mF}{\ensuremath{\mat{F}}}

\providecommand{\mI}{\ensuremath{\mat{I}}}

\providecommand{\mM}{\ensuremath{\mat{M}}}
\providecommand{\mN}{\ensuremath{\mat{N}}}
\providecommand{\mO}{\ensuremath{\mat{O}}}
\providecommand{\mP}{\ensuremath{\mat{P}}}

\providecommand{\mX}{\ensuremath{\mat{X}}}
\providecommand{\mY}{\ensuremath{\mat{Y}}}

\newcommand{\m}{\boldsymbol}
\allowdisplaybreaks[4]
\usepackage[colorlinks = true,
linkcolor = blue,
urlcolor  = blue,
citecolor = blue,
anchorcolor = blue]{hyperref}

\newcommand{\mc}[1]{\mathcal{#1}}

\usepackage[margin=0.75in]{geometry}
\title{\vspace{0.8cm}\LARGE \bf Simultaneous Sensor and Actuator Selection/Placement through Output Feedback Control}
\author{Sebastian Nugroh$\text{o}^*$, Ahmad F. Tah$\text{a}^*$, Tyler Summer$\text{s}^{\dagger}$, Nikolaos Gatsi$\text{s}^*$
	\thanks{
		*Department of Electrical and Computer Engineering, The University of Texas at San Antonio, 1 UTSA Circle, San Antonio, TX 78249.
		$^\dagger$Department of Mechanical Engineering, The University of Texas at Dallas, 800 W Campbell Rd, Richardson, TX 75080.
		Emails: sebastian.nugroho@my.utsa.edu, \{ahmad.taha, nikolaos.gatsis\}@utsa.edu, tyler.summers@utdallas.edu. This material is based upon work supported by the National Science Foundation under Grants CMMI 1728629 and CMMI 1728605.}
}

\begin{document}

\maketitle
\thispagestyle{empty}
\pagestyle{empty}

\begin{abstract}
In most dynamic networks, it is impractical to measure all of the system states; instead, only a subset of the states are measured through sensors. Consequently, and unlike full state feedback controllers, output feedback control utilizes only the measured states to obtain a stable closed-loop performance. This paper explores the interplay between the selection of minimal number of sensors and actuators (SaA) that yield a stable closed-loop system performance. Through the formulation of the static output feedback control problem, we show that the simultaneous selection of minimal set of SaA is a combinatorial optimization problem with mixed-integer nonlinear matrix inequality constraints. To address the computational complexity, we develop two approaches: The first approach relies on integer/disjunctive programming principles, while the second approach is a simple algorithm that is akin to binary search routines. The optimality of the two approaches is also discussed. Numerical experiments are included showing the performance of the developed approaches.
\end{abstract}

\vspace{-0.3cm}
\section{Introduction}
The interplay between the selection of minimal number of sensors and actuators (SaA) in dynamic systems is investigated in this paper. In general, the SaA selection or placement problem can be described as finding the optimal binary, on/off configuration of SaA that satisfy certain dynamic system metrics such as closed-loop system stability, output-feedback stability, linear quadratic regulator and robust $\mc{H}_{2}$/$\mc{H}_{\infty}$ control/estimation metrics. This problem has potential applications in areas such as: large scale power systems \cite{Taylor2017,pequito2013framework}, power systems integration with microgrids \cite{Bansal2016}, municipal water networks \cite{berry2005sensor}, and transportation systems \cite{TubaishatWCM616,Contreras2016}.

Various studies investigate the problem of selecting sensors or actuators separately, while invoking the separation principle that decouples the problems of designing controllers and state estimators, while assuming classical state feedback controller. A more interesting problem is that of simultaneously selecting SaA in the context of output feedback control, where the control law is obtained explicitly from the output measurements, rather than the states of the network. Even when the separation principle is conveniently exploited, the SaA selection problems are inherently coupled. 

Three major approaches have been developed in the recent literature of SaA selection. The first approach is based on combinatorial algorithms, heuristics, and detailed algorithms that often exploit network structure and properties~\cite{tzoumas2016,zhang2017sensor,nepusz2012controlling,ruths2014control,olshevsky2014minimal,pequito2016,Summers2016,summers2016actuator,Haber2017}. The second approach entails utilizing semidefinite programming (SDP) formulations of control/estimation methods while including sparsity promoting penalties on the gain matrix---thereby minimizing the total number of activated SaA~\cite{Polyak-LMI_sparse_fb,Dhingra2014,Argha2016,Munz2014}. The third approach uses a combination of mixed-integer convex programming, convex relaxations and approximations to obtain the minimal set of SaA~\cite{Taylor2017,Chanekar2017,Taha2017d}. In particular, the problem of simultaneously selecting/placing SaA with dynamic output feedback control is studied in~\cite{de2000linear,Argha2016}. 
In this paper, we investigate the problem of simultaneously selecting SaA through static output feedback control framework, where the objective is to stabilize the closed-loop system through the least number of SaA given logistic constraints on the selection of SaA. Two different approaches to solve this problem are proposed. 

The paper organization are as follows. First, we discuss the needed assumptions, definitions, and the formulation of the classical static output feedback problem through an SDP---all in Section~\ref{sec:PbmForm}. The problem formulation is presented in Section~\ref{sec:pbmform}, where we show that the simultaneous SaA selection requires solving a nonconvex optimization problem with mixed-integer nonlinear matrix inequality (MI-NMI) constraints. Section~\ref{sec:Big-M} presents the first approach, whereby the problem is transformed to MI-SDP by using disjunctive programming principles~\cite{nemhauser1988integer,grossmann2002review}. 
Section~\ref{sec:BSAlgo} presents a departure from the mixed-integer formulations to an algorithm that is akin to binary search routines. The developed algorithm leverages the SaA problem structure and the suboptimality or infeasibility of specific SaA combinations. 
We prove that both approaches yield optimal solutions to the formulated nonconvex problem. Numerical tests are provided in Section~\ref{sec:results}. 

Some of the mathematical proofs are omitted in this version of the paper, but will be included in an extended version of this work. 
\section{Static Output Feedback Control Review and Problem Formulation}~\label{sec:PbmForm}

\vspace{-0.4cm}
In this section, we present some necessary background including the definition of static output feedback stabilizability and the SDP that solves for the output feedback gain given a fixed SaA combination.

\vspace{-0.3cm}
\subsection{Notation}
The set of $n\times n$ symmetric and positive definite matrices are denoted $\mathbb{S}^{n}$ and $\mathbb{S}^{n}_{++}$. For a square matrix $\m X$, the notation $\Lambda(\m X)$ denotes the set of all eigenvalues of $\m X$. The function $\mathrm{Re}(c)$ extracts the real part of a complex number $c$, whereas $\mathrm{blkdiag}(\cdot)$ is used to construct a block diagonal matrix.  For a matrix $\mX\in\mathbb{R}^{p\times q}$, the operator $\mathrm{Vec}(\mX)$ returns a stacked $pq\times 1$ column vector of entries of $\mX$, while $\mathrm{Diag}(\mY)$ returns a $n\times 1$ column vector of diagonal entries of square matrix $\mY\in\mathbb{R}^{n\times n}$. The symbol $\otimes$ denotes the Kronecker product. For any $x\in\mathbb{R}$, $\vert x\vert$ and $\lceil x \rceil$ denote the absolute value and ceiling function of $x$. The cardinality of a set $\mathcal{S}$ is denoted by $\vert\mathcal{S}\vert$, whereas $(0)^n$ denotes a $n$-tuple with zero valued elements.
 
\vspace{-0.3cm}
\subsection{Systems Description}
Consider a linear time invariant (LTI) dynamical system consisting of $N$ nodes, with $\mathcal{N} = \lbrace 1,\ldots,N\rbrace$ defining the set of nodes, modeled in the following state-space equations
\vspace{-0.4cm}
\begin{subequations}~\label{equ:StateSpace-all}
\begin{align}
\m{\dot{x}}(t) &= \m{Ax}(t)+\boldsymbol{Bu}(t)~\label{equ:StateSpace-x}  \\
\m{y}(t) &= \m{Cx}(t),~\label{equ:StateSpace-y} 
\end{align}
\end{subequations}
where the state, input, and output vectors on each node $i \in \mathcal{N}$ are represented by $\m{x}_i(t) \in \mathbb{R}^{n_{x_i}}$, $\m{u}_i(t) \in \mathbb{R}^{n_{u_i}}$, and $\m{y}_i(t) \in \mathbb{R}^{n_{y_i}}$. The global state, input, and output vectors are written as $\m{x}(t) \triangleq [\m{x}_1^{\top}(t),\ldots,\m{x}_N^{\top}(t)]^{\top}$, $\m{u}(t) \triangleq [\m{u}_1^{\top}(t),\ldots,\m{u}_N^{\top}(t)]^{\top}$, and $\m{y}(t) \triangleq [\m{y}_1^{\top}(t),\hdots,\m{y}_N^{\top}(t)]^{\top}$ where $\m{x}(t)\in \mathbb{R}^{n_{x}}$, $\m{u}(t)\in \mathbb{R}^{n_{u}}$, and $\m{y}(t)\in \mathbb{R}^{n_{y}}$. We assume that the SaA on each node $i$ only correspond to that particular node. Therefore, $\m B$ and $\m C$ can be respectively constructed as $\m{B} \triangleq \mathrm{blkdiag}(\m{B}_{1},\m{B}_{2},\ldots,\m{B}_{N})$ and $\m{C} \triangleq \mathrm{blkdiag}(\m{C}_{1},\m{C}_{2},\ldots,\m{C}_{N})$ where $\m B \in \mathbb{R}^{n_x\times n_u} $ and $\m C \in \mathbb{R}^{n_y\times n_x} $. This assumption enforces the coupling among nodes to be represented in the state evolution matrix $\m{A} \in \mathbb{R}^{n_x\times n_x} $, which is realistic in various dynamic networks as control inputs and observations are often determined locally. In addition, we also assume that $\m B$ and $\m C$ are full column rank and full row rank, respectively.

To formalize the SaA selection problem, let $\gamma_i\in \lbrace0,1\rbrace$ and $\pi_i \in \lbrace0,1\rbrace$ be two binary variables that represent the selection of SaA at node $i$ of the dynamic network. We consider that $\gamma_i = 1$ if the sensor of node $i$ is selected (or activated) and $\gamma_i = 0$ otherwise. Similarly, $\pi_i = 1$ if the actuator of node $i$ is selected and $\pi_i = 0$ otherwise. The augmented dynamics can be written as
\vspace{-0.00cm}
\begin{subequations}\label{equ:StateSpaceSaA-all}
\begin{align}
\m{\dot{x}}(t) &= \m{Ax}(t)+\m{B\Pi u}(t) ~\label{equ:StateSpaceSaA-x}\\
\m{y}(t) &= \m{\Gamma Cx}(t), ~\label{equ:StateSpaceSaA-y} 
\end{align}
\end{subequations}
where $\m{\Pi}$ and $\m{\Gamma}$ are symmetric block matrices defined as
\begin{subequations}\label{equ:defPiGamma}
\begin{align}
\m{\Pi} &\triangleq \mathrm{blkdiag}(\pi_1\m{I}_{n_{u_{1}}},\pi_2\m{I}_{n_{u_{2}}},\ldots,\pi_N\m{I}_{n_{u_{N}}})\label{equ:defPi}\\
\m{\Gamma}&\triangleq\mathrm{blkdiag}(\gamma_1\m{I}_{n_{y_{1}}},\gamma_2\m{I}_{n_{y_{2}}},\ldots,\gamma_N\m{I}_{n_{y_{N}}}).\label{equ:defGamma}
\end{align} 
\end{subequations}

\vspace{-0.3cm}
\subsection{The Static Output Feedback Stabilizability Problem}

We begin this section by providing the definition of static output feedback stabilizability.

\begin{mydef}~\label{def:OutputFeedbackStabilizability}
The dynamical system \eqref{equ:StateSpace-all} is stabilizable via static output feedback if there exists $\m F \in \mathbb{R}^{n_u\times n_y}$, with control law defined as $\m u(t) = \m F\m y(t)$, such that $\mathrm{Re}(\lambda) < 0$ for every $\lambda \in \Lambda(\m{A}+\m{B}\m{ F C})$.
\end{mydef}

By using the above definition, the static output feedback stabilizability problem can be defined as the problem of finding $\m F$ such that the closed loop system $\m{A}+\m{B}\m{ F C}$ is asymptotically stable. Throughout this paper, we require that dynamical system \eqref{equ:StateSpace-all} satisfies the following assumption.

\begin{asmp}~\label{asmp:detectandstable}
The following conditions apply to \eqref{equ:StateSpace-all}:
\begin{enumerate}
\item The pair $(\m A, \m B)$ is stabilizable,
\item The pair $(\m A, \m C)$ is detectable.
\end{enumerate}
\end{asmp}

Note that above assumption is not enough to guarantee that \eqref{equ:StateSpace-all} is stabilizable via static output feedback. To proceed, the following proposition provides a sufficient condition for static output feedback stabilizability.
\begin{myprs}~\label{prs:OFS}
The dynamic network \eqref{equ:StateSpace-all} is static output feedback stabilizable if there exist an invertible matrix $\m M \in \mathbb{R}^{n_u\times n_u}$, matrices $\m P \in \mathbb{S}^{n_x}_{++}$, $\m N \in \mathbb{R}^{n_u\times n_y}$, and $\mF \in \mathbb{R}^{n_u\times n_y}$ such that the following linear matrix inequalities are feasible
\begin{subequations}\label{eq:OFS}
\begin{align}
&\m{A}^{\top}\m P+\m{PA}+\m{C}^{\top}\m N^{\top}\m B^{\top}+\m{B}\m{N C} 
\prec 0~\label{eq:OFS1} \\
&\m B \m M = \m P \m B, ~\label{eq:OFS2} 
 \end{align}
 \end{subequations}
with control law $\m u (t) = \mF \m y(t)$ where $\m{F}=\m{M}^{-1}\m{N}$. 
\end{myprs}

The proof of the above proposition is available in \cite{Crusius1999}. The condition presented in Proposition \ref{prs:OFS} allows the static output feedback stabilization problem to be solved as an LMI. The problem formulation of output feedback stabilizability with simultaneous SaA selection is given next.

\section{Problem Formulation}~\label{sec:pbmform}
The simultaneous SaA selection with static output feedback control is the problem of selecting a minimal set of SaA while still maintaining the stability of the system through static output feedback control. Thus, based on Proposition~\ref{prs:OFS}, the SaA selection problem for output feedback stabilization can be formulated as follow.
\begin{subequations}~\label{eq:OFS-SaAProblem}
\begin{align}
\minimize\;\; & \sum_{k=1}^{N}  \pi_k +  \gamma_k ~\label{eq:OFS-SaAProblem-1} \\
\subjectto \;\; &\m{A}^{\top}\m P+\m{PA}+\m{C}^{\top}\m\Gamma\m N^{\top}\m\Pi\m B^{\top} \nonumber \\ 
&+\m{B}\m\Pi\m{N \Gamma C}\prec 0 ~\label{eq:OFS-SaAProblem-2}\\
& \m{B}\m{\Pi M}=\m{P} \m B\m\Pi ~\label{eq:OFS-SaAProblem-3}\\
&\m \Phi \bmat{\m \pi \\ \m \gamma} \leq \m \phi ~\label{eq:OFS-SaAProblem-4}\\
& \m{P} \succ 0,\;\m\pi \in \{0,1\}^N,\;\m\gamma \in \{0,1\}^N ~\label{eq:OFS-SaAProblem-5}. 
\end{align}
\end{subequations}

In \eqref{eq:OFS-SaAProblem}, the optimization variables are $\lbrace\m\pi,\m\gamma,\m N,\m M,\m P\rbrace$ with $\mP\in\mathbb{S}^{n_x}$, $\m\pi = [\pi_1,\ldots,\pi_N]^{\top}$, and $\m\gamma = [\gamma_1,\ldots,\gamma_N]^{\top}$.
The additional constraint \eqref{eq:OFS-SaAProblem-4} can be regarded as a linear logistic constraint, which is useful to model preferred activation or deactivation of SaA on particular nodes and to define the desired minimum and maximum number of active SaA. This constraint is also useful in multi-period selection problems where certain actuators and sensors are deactivated due to logistic constraints. 

Upon solving \eqref{eq:OFS-SaAProblem}, the SaA selection is obtained and represented by $\lbrace\m \pi^*,\,\m \gamma^*\rbrace$ with static output feedback gain $\mF$ to be computed as $\m{M}^{-1}\m{N}$, assuming that $\m{M}$ is invertible. Note that, \eqref{eq:OFS-SaAProblem} is nonconvex due to the presence of MI-NMI in the form of $\m\Pi\m{N \Gamma}$ and mixed-integer bilinear matrix equality in \eqref{eq:OFS-SaAProblem-3}. Thus, problem \eqref{eq:OFS-SaAProblem} cannot be solved by any general-purpose mixed integer convex programming solver. To that end, two different approaches that solve or approximate \eqref{eq:OFS-SaAProblem} are developed. The first approach is based on disjunctive programming, while the other approach is based on a binary search algorithm. The next section presents the first approach.

\section{Disjunctive Programming for SaA Selection}~\label{sec:Big-M}

\vspace{-0.4cm}
The first approach is developed based on disjunctive programming principles~\cite{nemhauser1988integer,grossmann2002review}. The following theorem presents this result. 

\begin{theorem}~\label{thrm:Big-M}
The optimization problem \eqref{eq:OFS-SaAProblem} is equivalent to
\vspace{-0.2cm}
\begin{subequations}~\label{eq:Big-M-SaAProblem}
\begin{align}
\minimize & \;\;\;\sum_{k=1}^{N} \pi_k + \gamma_k ~\label{eq:Big-M-SaAProblem-1} \\
\subjectto  & \notag\\
& \hspace{-1.3cm}\m{A}^{\top}\m{P}+\m{PA}+\m{C}^{\top}\m\Theta^{\top}\m B^{\top}+\m{B}\m {\Theta C} \prec 0 ~\label{eq:Big-M-SaAProblem-2}\\
&\hspace{-1.3cm}\m \Omega = (\mB^{\top}\mB)^{-1}\mB^{\top}\mP\mB~\label{eq:Big-M-SSAProblem-new1} \\
				&\hspace{-1.3cm}\m \Xi = (\mI-\mB(\mB^{\top}\mB)^{-1}\mB^{\top})\mP\mB~\label{eq:Big-M-SSAProblem-new4} \\
				&\hspace{-1.3cm}\bmat{\m\Psi_1(\mN,\m\Theta)\\\m\Psi_2(\mM,\m\Omega)\\\m\Psi_3(\m\Xi)} \leq \bmat{L_1\m\Delta_1(\m \Gamma,\m\Pi)\\L_2\m\Delta_2(\m\Pi)\\L_3\m\Delta_3(\m\Pi)}~\label{eq:Big-M-SSAProblem-all} \\
&\hspace{-1.3cm} \eqref{eq:OFS-SaAProblem-4},\,\eqref{eq:OFS-SaAProblem-5},~\label{eq:Big-M-SaAProblem-4}
\end{align}
\end{subequations}
where 
	{\small \begin{subequations}
		\begin{align*}
							\m\Psi_1(\mN,\m\Theta) = &\bmat{\mathrm{Vec}(\m\Theta)\\-\mathrm{Vec}(\m\Theta)\\\mathrm{Vec}(\m\Theta)\\-\mathrm{Vec}(\m\Theta)\\\mathrm{Vec}(\m\Theta-\mN)\\-\mathrm{Vec}(\m\Theta-\mN)} \\
				\m\Delta_1(\m \Gamma,\m\Pi) =  &\bmat{\mathrm{Diag}(\mI_{n_y}\otimes \m\Pi)\\\mathrm{Diag}(\mI_{n_y}\otimes \m\Pi)\\\mathrm{Diag}(\m\Gamma\otimes\mI_{n_u})\\\mathrm{Diag}(\m\Gamma\otimes\mI_{n_u})\\\mathrm{Diag}(2\mI_{n_u\times n_y}-\mI_{n_y}\otimes \m\Pi-\m\Gamma\otimes\mI_{n_u})\\\mathrm{Diag}(2\mI_{n_u\times n_y}-\mI_{n_y}\otimes \m\Pi-\m\Gamma\otimes\mI_{n_u})} \\
				\m\Psi_2(\mM,\m\Omega)  = &\bmat{\mathrm{Vec}(\mM)\\-\mathrm{Vec}(\mM)\\\mathrm{Vec}(\m\Omega)\\-\mathrm{Vec}(\m\Omega)\\\mathrm{Vec}(\mM-\m\Omega)\\-\mathrm{Vec}(\mM-\m\Omega)} \\
				\m\Delta_2(\m\Pi) =  &\bmat{\mathrm{Diag}(\mI_{n_u^2}-\mI_{n_u}\otimes \m\Pi+\m\Pi\otimes\mI_{n_u})\\\mathrm{Diag}(\mI_{n_u^2}-\mI_{n_u}\otimes \m\Pi+\m\Pi\otimes\mI_{n_u})\\\mathrm{Diag}(\mI_{n_u^2}+\mI_{n_u}\otimes \m\Pi-\m\Pi\otimes\mI_{n_u})\\\mathrm{Diag}(\mI_{n_u^2}+\mI_{n_u}\otimes \m\Pi-\m\Pi\otimes\mI_{n_u})\\\mathrm{Diag}(2\mI_{n_u^2}-\mI_{n_u}\otimes \m\Pi-\m\Pi\otimes\mI_{n_u})\\\mathrm{Diag}(2\mI_{n_u^2}-\mI_{n_u}\otimes \m\Pi-\m\Pi\otimes\mI_{n_u})} \\
				\m\Psi_3(\m \Xi) = &\bmat{\mathrm{Vec}(\m\Xi)\\-\mathrm{Vec}(\m\Xi)} \\
			\m\Delta_3(\m\Pi) =  &\bmat{\mathrm{Diag}(\mI_{n_x\times n_u}-\m\Pi\otimes\mI_{n_x})\\\mathrm{Diag}(\mI_{n_x\times n_u}-\m\Pi\otimes\mI_{n_x})},
		\end{align*}
	\end{subequations}}

with $\m \Theta \in \mathbb{R}^{n_u\times n_y}$, $\m \Omega \in \mathbb{R}^{n_u\times n_u}$ and $\m \Xi \in \mathbb{R}^{n_x\times n_u}$ are additional variables and $L_1,L_2,L_3\in\mathbb{R}_{++}$ are sufficiently large constants.
\end{theorem}

The proof is omitted from this version of the work, and will be included in the extended version of the manuscript \cite{Nugroho2018}.
\vspace{-0.05cm}
Although \eqref{eq:Big-M-SaAProblem} is equivalent to \eqref{eq:OFS-SaAProblem}, the quality of the solution that comes out of \eqref{eq:Big-M-SaAProblem} depends on the choice of $L_1$ and $L_2$. Theorem \ref{thrm:Big-M} allows the SaA selection for static output feedback stabilizability to be solved as a MI-SDP. The next section presents a departure from MI-SDP to an algorithm that solves~\eqref{eq:OFS-SaAProblem}.

\section{Binary Search Algorithm for SaA Selection}~\label{sec:BSAlgo}

\vspace{-0.91cm}
\subsection{Introduction}
In this section, we present an algorithm that is similar in spirit to binary search routines. In what follows, we provide the definitions and examples that are important to understand the algorithm.

\begin{mydef}~\label{def:setS}
Let $\mathcal{S}_\pi$ and $\mathcal{S}_\gamma$ be two $N$-tuples representing the selection of actuator and sensor. That is, $\mathcal{S}_\pi \triangleq (\pi_1,\ldots,\pi_N)$ and $\mathcal{S}_\gamma \triangleq (\gamma_1,\ldots,\gamma_N)$. Then, the selection of SaA can be defined as $\mathcal{S} \triangleq (\mathcal{S}_\pi,\mathcal{S}_\gamma)$ such that $\lbrace\m \Pi,\,\m \Gamma\rbrace = \mathcal{G}(\mathcal{S})$, $\m\Pi = \mathcal{G}_{\pi}(\mathcal{S})$, and $\m\Gamma = \mathcal{G}_{\gamma}(\mathcal{S})$ where $\mathcal{G}(\cdot):\mathcal{S}\rightarrow \mathbb{R}^{n_u\times n_u}\times \mathbb{R}^{n_y\times n_y}$, $\mathcal{G}_{\pi}(\cdot):\mathcal{S}\rightarrow \mathbb{R}^{n_u\times n_u}$, and $\mathcal{G}_{\gamma}(\cdot):\mathcal{S}\rightarrow \mathbb{R}^{n_y\times n_y}$ are linear maps. The number of nodes with active SaA can be defined as $\mathcal{H}(\mathcal{S}) \triangleq \sum_{k=1}^{N}  \pi_k +  \gamma_k$ where $\mathcal{H}(\cdot):\mathcal{S}\rightarrow \mathbb{Z}_{+}$.
\end{mydef}

\begin{mydef}~\label{def:bigS}
Let $\mathbfcal{S}\triangleq \lbrace\mathcal{S}_q\rbrace_{q=1}^\sigma$ be the candidate set such that it contains all possible combinations of SaA where $\sigma$ denotes the number of total combinations, i.e., $\sigma\triangleq\vert\mathbfcal{S}\vert$. Then, the following conditions hold: 
\begin{enumerate}
\item For all $\mathcal{S}\in\mathbfcal{S}$, $\lbrace\m \Pi,\,\m \Gamma\rbrace = \mathcal{G}(\mathcal{S})$ is feasible for \eqref{eq:OFS-SaAProblem-4}, and
\item $\mathbfcal{S}$ is ordered such that $\mathcal{H}(\mathcal{S}_{q-1}) \leq \mathcal{H}(\mathcal{S}_q)$.
\end{enumerate}
\end{mydef}

\begin{exmpl}\label{ex:Ex1}
Suppose that the dynamical system consists of two nodes with one input and one output on each node. If the logistic constraint dictates that $ 1 \leq \mathcal{H}(\mathcal{S}) <  4$ for all $\mathcal{S}\in\mathbfcal{S}$, then the candidate set $\mathbfcal{S}$ can be constructed as
\vspace{-0.1cm}
\begin{align*}
\mathbfcal{S} = \Big\lbrace
&(1,0,0,0),(0,1,0,0),(0,0,1,0),(0,0,0,1),\\
&(1,1,0,0),(1,0,1,0),(1,0,0,1),(0,1,1,0),\\
&(0,1,0,1),(0,0,1,1),(1,1,1,0),(1,1,0,1),\\
&(1,0,1,1),(0,1,1,1)\Big\rbrace.
\end{align*}
\end{exmpl} 

\vspace{-0.6cm}
\subsection{Binary Search Algorithm to Solve \eqref{eq:OFS-SaAProblem}}\label{BSA-1}

The objective of this algorithm is to find an optimal solution $\mathcal{S}^*\in\mathbfcal{S}$ such that $\mathcal{H}(\mathcal{S}^*)\leq\mathcal{H}(\mathcal{S})$ for all $\mathcal{S}\in\mathbfcal{V}$ where $\mathbfcal{V}\triangleq\lbrace \mathcal{S}\in\mathbfcal{S} \,|\, \lbrace\m \Pi,\,\m \Gamma\rbrace = \mathcal{G}(\mathcal{S}) \textit{ is feasible for \eqref{eq:OFS}}\rbrace$. Realize that $\mathcal{S}^*$ might be not unique\footnote[1]{The solution might not be unique since there could be more than one combinations of SaA that yield minimum number of activated SaA, while still generating feasible solution to the LMIs for static output feedback stabilizability.} and finding one is adequate for our purpose.

The routine to solve SaA selection with static output feedback is now described as follows. Let $p$ be the index of iteration and $q$ be the index of  position in the ordered set $\mathbfcal{S}$. Hence at iteration $p$, the candidate set that contains all possible combinations of SaA can be represented as $ \mathbfcal{S}_p$, with $\sigma=\vert\mathbfcal{S}_p\vert$, and any element of $ \mathbfcal{S}_p$ at position $q$ can be represented by $\mathcal{S}_q$. Also, let $\mathcal{S}^*$ be the current solution, which is initialized as $\mathcal{S}^* = (0)^{2N}$.

Next, obtain $\mathcal{S}_q$ where $\mathcal{S}_q\in\mathbfcal{S}_p$ and $q = \lceil\sigma/2\rceil$. At this step, we need to determine whether system \eqref{equ:StateSpaceSaA-all} is output feedback stabilizable with a certain combination of SaA $\lbrace\m \Pi_q,\,\m \Gamma_q\rbrace = \mathcal{G}(\mathcal{S}_q)$. To that end, we use the LMIs from Proposition \ref{prs:OFS}. When solving \eqref{eq:OFS} for given $\lbrace\m \Pi_q,\,\m \Gamma_q\rbrace$, let $\m B$ and $\m C$ in \eqref{eq:OFS} be substituted with $\m B_q$ and $\m C_q$ so that both represent the nonzero components of $\m B \m \Pi_q$ and $\m \Gamma_q\m C$ that correspond to activated SaA. If $\m B_q$ and $\m C_q$ are feasible for \eqref{eq:OFS}, then $\mathcal{S}^*$ is updated such that $\mathcal{S}^*=\mathcal{S}_q$. Since $\mathcal{S}_q$ is feasible, then we can discard all combinations that have more or equal number of active SaA. Otherwise, if $\m B_q$ and $\m C_q$ are infeasible for \eqref{eq:OFS}, then we can discard $\mathcal{S}_q$ and all combinations that (a) have less number of active SaA than $\mathcal{S}_q$ \textbf{and} (b) the active SaA are included in $\mathcal{S}_q$. 

Realize that the above method reduces the size of $\mathbfcal{S}_p$ in every iteration because one or more elements of $\mathbfcal{S}_p$ are discarded. Let $\mathbfcal{S}_{p+1}$ be the new set of all possible combination of SaA after all unwanted combinations of SaA are discarded. Then, we can update the number of possible combinations of SaA as $\sigma = \vert\mathbfcal{S}_{p+1}\vert$. The algorithm now continues and terminates when $\mathbfcal{S}_p = \emptyset$. The detail of this algorithm is given in Algorithm~\ref{alg:DaC}. Example \ref{exp:2} gives an illustration how $\mathbfcal{S}_p$ is constructed in every iteration.

\vspace{-0.05cm}
\begin{exmpl}\label{exp:2}
Consider again the  dynamic system from Example~\ref{ex:Ex1}. Let $(1,0,0,1)$ be the starting combination and, for the sake of illustration, assume that \eqref{eq:OFS} is infeasible for this combination. Then, by Algorithm \ref{alg:DaC}, combinations $(1,0,0,0)$ and $(0,0,0,1)$ are discarded. The candidate set now comprises the following elements 
	\begin{align*}
	\mathbfcal{S}_2= \Big\lbrace
&(0,1,0,0),(0,0,1,0),
(1,1,0,0),(1,0,1,0),\\
&(0,1,1,0),(0,1,0,1),(0,0,1,1),(1,1,1,0),\\&(1,1,0,1),(1,0,1,1),(0,1,1,1)\Big\rbrace.
	\end{align*}
Let $(0,1,0,1)$ be the new starting point and assume that this combination is feasible for \eqref{eq:OFS}. Then, all combinations that have greater or equal number of active SaA can be discarded. The remaining possible candidates on the candidate set are
	\begin{align*}
\mathbfcal{S}_3 = \Big\lbrace
&(0,1,0,0),(0,0,1,0)\Big\rbrace.
\end{align*}
This algorithm continues in a fashion similar to the above routine. If none of these combinations in $\mathbfcal{S}_3$ is feasible, then Algorithm~\ref{alg:DaC} returns $\mc{S}^*= (0,1,0,1)$ as the solution.  
\end{exmpl} 

\vspace{-0.2cm}
\begin{algorithm}[h]
\caption{Binary Search Algorithm}\label{alg:DaC}
\begin{algorithmic}[1]
\STATE \textbf{initialize:} $\mathcal{S}^*=(0)^{2N}$, $p=1$
\STATE \textbf{input:}  $\mathbfcal{S}_p$
\WHILE{$\mathbfcal{S}_p \neq \emptyset$}
\STATE \textbf{compute:} $\sigma\leftarrow\vert\mathbfcal{S}_p\vert$, $q \leftarrow \lceil \sigma/2 \rceil$, $\mathcal{S}_q\in\mathbfcal{S}_p$
\IF{\eqref{eq:OFS} is feasible}\label{alg1:step5}
\STATE $\mathcal{S}^*\leftarrow\mathcal{S}_q$, $\mathbfcal{S}_p \leftarrow \mathbfcal{S}_p\setminus\lbrace\mathcal{S}\in\mathbfcal{S}_p\,\vert\,\mathcal{H}(\mathcal{S}) \geq \mathcal{H}(\mathcal{S}_q)\rbrace$
\ELSE
\STATE $\mathbfcal{S}_p\leftarrow\mathbfcal{S}_p\setminus\lbrace\mathcal{S}\in\mathbfcal{S}_p\,\vert\,\mathcal{S}_q\vee\mathcal{S} = \mathcal{S}_q \rbrace$ \label{alg1:step9}
\ENDIF
\STATE $p \leftarrow p + 1$
\ENDWHILE 
\STATE \textbf{output:} $\mathcal{S}^*$
\end{algorithmic}
\end{algorithm} 
\vspace{-0.1cm}

\begin{theorem}\label{thm:thm2}
Algorithm~\ref{alg:DaC} returns an optimal solution of~\eqref{eq:OFS-SaAProblem}.
\end{theorem}
The proof is omitted from this version of the work, and will be included in the extended version of the manuscript \cite{Nugroho2018}. The reason why Algorithm~\ref{alg:DaC} returns \textit{an} optimal solution of \eqref{eq:OFS-SaAProblem} is due to the fact that two SaA configurations can return the same objective value of~\eqref{eq:OFS-SaAProblem}. However, one SaA configuration can yield a \textit{more} stable closed loop system in terms of the distance from the $j\omega$-axis. This is shown in the numerical tests (Section~\ref{sec:results}).

\vspace{-0.3cm}
\subsection{Modified Binary Search Algorithm}\label{BSA-2}

In Algorithm \ref{alg:DaC}, the LMI \eqref{eq:OFS} is solved in every iteration to determine whether a particular combination of SaA yields a feasible or infeasible solution to the static output feedback problem. In this section, we provide a modification to Algorithm~\ref{alg:DaC} so that it no longer requires solving the LMI feasibility problem at each iteration---potentially resulting in a reduction in the computational time. 

This simple modification is carried out by replacing Step \ref{alg1:step5} in Algorithm \ref{alg:DaC} with stabilizability and detectability tests of linear dynamic systems. To this end, the following propositions are useful. 

\iftrue
\begin{myprs}\label{prs:RedundantControllability}
Let  $\mathcal{S}$ be an arbitrary combination of SaA. If system \eqref{equ:StateSpaceSaA-all} is stabilizable for $\mathcal{S}$, then activating one or more actuators from $\mathcal{S}$ will keep system \eqref{equ:StateSpaceSaA-all} stabilizable. Similarly, if system \eqref{equ:StateSpaceSaA-all} is unstabilizable for $\mathcal{S}$, then deactivating one or more actuators from $\mathcal{S}$ will keep system \eqref{equ:StateSpaceSaA-all} unstabilizable.
\end{myprs}

\begin{proof}
We first prove the first part of the proposition. Let $\mathcal{S}\in\mathbfcal{S}$ with $\m \Pi=\mathcal{G}_{\pi}(\mathcal{S})$. Let $\mB_1\in\mathbb{R}^{n_x\times m_1}$ be a matrix that represents the nonzero components of $\mB\m\Pi$ that correspond to activated actuators. Since the pair $(\mA,\mB_1)$ is stabilizable, then we have $\m v^{\top}\mB_1 = \m u$ where $\m u^{\top}\in\mathbb{R}^{m_1}$ and $\m u\neq \m0$ \cite[Theorem 1]{Hautus1970} for all $\m v\in\lbrace \m v\in\mathbb{R}^{n_x}\,\vert\,\m v^{\top}(\mA-\lambda\mI)=0,\,\forall\lambda\geq 0,\,\lambda\in\Lambda(\mA)\rbrace$. Now, define $\mB \in \mathbb{R}^{n_x\times m}$ and $\mB_2 \in \mathbb{R}^{n_x\times m_2}$, with $m=m_1+m_2$ and $m \leq n_u$, such that $\mB_2$ represents the addition of activated actuators and $\mB = \bmat{\mB_1&\mB_2}$. Then, 
\vspace{-0.1cm}
\begin{align*}
\m v^{\top}\mB =\m v^{\top}\bmat{\mB_1&\mB_2} = \bmat{\m v^{\top}\mB_1&\m v^{\top}\mB_2} = \bmat{\m u&\m v^{\top}\mB_2}.
\end{align*}
Since $\m u \neq \m0$, the pair $(\mA,\mB)$ is also stabilizable, proving the first part of the proposition. Since the second part of the proposition is the contraposition of the first part, then the proof is complete.
\end{proof}

\else

\begin{proof}
We first prove the first part of the proposition. Let $\mathcal{S},\bar{\mathcal{S}},\hat{\mathcal{S}}\in\mathbfcal{S}$. The statement is equivalent to saying that, if the pair $(\m A,\m B \m \Pi)$ is stabilizable with $\m \Pi=\mathcal{G}_{\pi}(\mathcal{S})$, then for any $\bar{\m \Pi}=\mathcal{G}_{\pi}(\bar{\mathcal{S}})$ where $\mathcal{S}_{\pi}\wedge\bar{\mathcal{S}}_{\pi} = \mathcal{S}_{\pi}$, the pair $(\m A,\m B \bar{\m \Pi})$ is also stabilizable. By contradiction, assume that the pair $(\m A,\m B \bar{\m \Pi})$ is unstabilizable. Then, according to the {Popov-Belevitch-Hautus} (PBH) criteria for stabilizability \cite{Ghosh1995}\cite{Hautus1970}, there exists $\lambda\geq0$ where $\lambda \in \Lambda(\m{A})$ such that $\mathrm{Rank}\bmat{\m A - \lambda \m I&\m B\bar{\m \Pi}}<n_x$. Since $\mathcal{S}_{\pi}\wedge\bar{\mathcal{S}}_{\pi} = \mathcal{S}_{\pi}$, then there exists $\hat{\m \Pi}=\mathcal{G}_{\pi}(\hat{\mathcal{S}}_{\pi})$ with $\hat{\m \Pi}\neq \mO$ such that $\bar{\m \Pi} = \m \Pi + \hat{\m \Pi}$ where $\mathcal{S}_{\pi}\vee\hat{\mathcal{S}}_{\pi}=\bar{\mathcal{S}}_{\pi}$, and $\mathcal{S}_{\pi}\wedge\hat{\mathcal{S}}_{\pi}=(0)^N$. Then, $\bmat{\m A - \lambda \m I&\m B\bar{\m \Pi}}$ can be written as
\vspace{-0.05cm}
\begin{align}
\bmat{\m A - \lambda \m I&\m B\bar{\m \Pi}} = \bmat{\m A - \lambda \m I&\m B\m \Pi}+ \bmat{\m O & \m B\hat{\m \Pi}}.\nonumber
\end{align}
\vspace{-0.05cm}
\noindent Because $\m B$ is full column rank and $\hat{\m \Pi}\neq \mO$, then by the property of matrix rank \cite[Section 0.4.2]{Horn2012}, $\mathrm{Rank}\bmat{\m A - \lambda \m I&\m B{\m \Pi}}<n_x$ for some $\lambda\geq0$ where $\lambda \in \Lambda(\m{A})$, i.e., the pair $(\m A,\m B \m \Pi)$ is unstabilizable. Since we have a contradiction, then we can conclude that the pair $(\m A,\m B \bar{\m \Pi})$ is stabilizable. 
The second part of the proposition is logically equivalent to the first part, and thus omitted here.
\end{proof}

\begin{myprs}\label{prs:RedundantControllability}
Let $\mathcal{S},\bar{\mathcal{S}},\hat{\mathcal{S}}\in\mathbfcal{S}$. If the pair $(\m A,\m B \m \Pi)$ is stabilizable with $\m \Pi=\mathcal{G}_{\pi}(\mathcal{S})$, then for any $\bar{\m \Pi}=\mathcal{G}_{\pi}(\bar{\mathcal{S}})$ where $\mathcal{S}^a_{\pi}\subset\bar{\mathcal{S}}^a_{\pi}$, the pair $(\m A,\m B \bar{\m \Pi})$ is also stabilizable. Similarly, if the pair $(\m A,\m B \m \Pi)$ is unstabilizable with $\m \Pi=\mathcal{G}_{\pi}(\mathcal{S})$, then for any $\hat{\m \Pi}=\mathcal{G}_{\pi}(\hat{\mathcal{S}})$ where $\hat{\mathcal{S}}^a_{\pi}\subset\mathcal{S}^a_{\pi}$, the pair $(\m A,\m B \hat{\m \Pi})$ is unstabilizable.
\end{myprs}

\begin{proof}
We first prove the first part of the proposition. By contradiction, assume that the pair $(\m A,\m B \bar{\m \Pi})$ is unstabilizable. Then, according to the {Popov-Belevitch-Hautus} (PBH) criteria for stabilizability \cite{Ghosh1995}\cite{Hautus1970}, there exists $\lambda\geq0$ where $\lambda \in \Lambda(\m{A})$ such that $\mathrm{Rank}\bmat{\m A - \lambda \m I&\m B\bar{\m \Pi}}<n_x$. Since $\mathcal{S}^a_{\pi}\subset\bar{\mathcal{S}}^a_{\pi}$, then there exists $\hat{\m \Pi}\in\hat{\mathcal{S}}_{\pi}$ with $\hat{\mathcal{S}}_{\pi}\in\mathbfcal{S}$ and $\hat{\m \Pi}\neq \mO$ such that $\bar{\m \Pi} = \m \Pi + \hat{\m \Pi}$ where $\mathcal{S}^a_{\pi}\,\cup\,\hat{\mathcal{S}}^a_{\pi}=\bar{\mathcal{S}}^a_{\pi}$, and $\mathcal{S}^a_{\pi}\,\cap\,\hat{\mathcal{S}}^a_{\pi}=\emptyset$. Then, $\bmat{\m A - \lambda_i \m I&\m B\bar{\m \Pi}}$ can be written as
\begin{align}
\bmat{\m A - \lambda \m I&\m B\bar{\m \Pi}} = \bmat{\m A - \lambda \m I&\m B\m \Pi}+ \bmat{\m O & \m B\hat{\m \Pi}}.\nonumber
\end{align}
\noindent Because $\m B$ is full column rank and $\hat{\m \Pi}\neq \mO$, then by the property of matrix rank \cite{Horn2012}, $\mathrm{Rank}\bmat{\m A - \lambda \m I&\m B{\m \Pi}}<n_x$ for some $\lambda\geq0$ where $\lambda \in \Lambda(\m{A})$, i.e., the pair $(\m A,\m B \m \Pi)$ is unstabilizable. Since we have a contradiction, then we can conclude that the pair $(\m A,\m B \bar{\m \Pi})$ is stabilizable. 
The second part of the proposition is logically equivalent to the first part, and thus omitted here.
\end{proof}
\fi 

\iftrue
\begin{myprs}\label{prs:RedundantObservability}
Let  $\mathcal{S}$ be an arbitrary combination of SaA. If system \eqref{equ:StateSpaceSaA-all} is detectable for $\mathcal{S}$, then activating one or more sensors from $\mathcal{S}$ will keep system \eqref{equ:StateSpaceSaA-all} detectable. Similarly, if system \eqref{equ:StateSpaceSaA-all} is undetectable for $\mathcal{S}$, then deactivating one or more sensors from $\mathcal{S}$ will keep system \eqref{equ:StateSpaceSaA-all} undetectable.
\end{myprs}
\else
\begin{myprs}\label{prs:RedundantObservability}
Let $\mathcal{S},\bar{\mathcal{S}},\hat{\mathcal{S}}\in\mathbfcal{S}$. If the pair $(\m A,\m\Gamma\m C)$ is detectable with $\m \Gamma=\mathcal{G}_{\gamma}(\mathcal{S})$, then for any $\bar{\m \Gamma}=\mathcal{G}_{\gamma}(\bar{\mathcal{S}})$ where $\mathcal{S}^a_{\gamma}\subset\bar{\mathcal{S}}^a_{\gamma}$, the pair $(\m A,\bar{\m\Gamma}\m C)$ is detectable. Similarly, if the pair $(\m A,\m\Gamma\m C)$ is undetectable with $\m \Gamma=\mathcal{G}_{\gamma}(\mathcal{S})$, then for any $\hat{\m \Gamma}=\mathcal{G}_{\gamma}(\hat{\mathcal{S}})$ where $\hat{\mathcal{S}}^a_{\gamma}\subset\mathcal{S}^a_{\gamma}$, the pair $(\m A,\hat{\m\Gamma}\m C)$ is undetectable.
\end{myprs}
\fi

Proposition \ref{prs:RedundantObservability} is the detectability equivalence of Propositions \ref{prs:RedundantControllability} and thus the proof is omitted for brevity. These two propositions allow discarding some combinations of SaA that are either unstabilizable and/or undetectable, or stabilizable and detectable but have more active SaA. Since stabilizability and detectability tests provide no guarantee of static output stabilizability for system \eqref{equ:StateSpace-all}, we save all combinations of SaA that pass the tests according to the routine in Algorithm~\ref{alg:DaC}. This allows the now-modified algorithm to consider the remaining combinations of SaA that contain more active SaA in the case when the best combination that passes the tests cannot give a stabilizing feedback gain. After all combinations of SaA that pass the tests have been stored, we solve \eqref{eq:OFS} starting from the combination with least number of SaA. If a feasible solution exists given this least-cost combination, the modified algorithm terminates. Otherwise, a stored combination having more active SaA is tested until \eqref{eq:OFS} is successfully solved.

The modified algorithm offers flexibility in assessing stabilizability/detectability of dynamic networks, although it no longer yields an optimal solution of~\eqref{eq:OFS-SaAProblem} as the stabilizability/detectability tests are not enough to guarantee the existence of stabilizing, static output feedback control gain. Specifically, either the PBH or the eigenvector tests \cite{Hautus1970} can be used. If the pairs $(\mA,\mB\m\Pi)$ and/or $(\mA,\m\Gamma\mC)$ have large condition number, then the eigenvector test is preferable to be used since MATLAB's rank function tends to return unreasonable results for pairs with large condition number \cite{olshevsky2014minimal}.

\begin{table*}[t]
\normalsize
	\centering
	\caption{ Numerical test results for the three scenarios/methods.}
	\begin{tabular}{|c|c|c|c|c|c|}
	\hline
		Scenario & {$\mathrm{Max}(\mathrm{Re}(\Lambda(\m A + \m B\m \Pi^*\mF \m \Gamma^*\m C)))$} & {$\sum_{k = 1}^{N} \pi_k+\gamma_k$} & {$\Delta t(s)$} & {$\mathrm{Iterations}$} & $\m{\gamma}^*$ and $\m{\pi}^*$ \\
		\hline \hline 
		\multirow{2}[1]{*}{MI-SDP} & \multirow{2}[1]{*}{-3.44$\times 10^{-3}$} & \multirow{2}[1]{*}{4} & \multirow{2}[1]{*}{14.13} & \multirow{2}[1]{*}{---} & $\m{\gamma}^* = \{0 ,    0   ,  1 ,    0 ,    1  ,   0  ,   0  ,   0   ,  0   ,  0\}$ \\
		&       &       &       &       & $\m{\pi}^* = \{1 ,    0  ,   0 ,    0   ,  1  ,   0  ,  0 ,    0  ,   0  ,   0\}$ \\
		\hline
		\multirow{2}[0]{*}{BSA-SDP} & \multirow{2}[0]{*}{-2.92$\times 10^{-3}$} & \multirow{2}[0]{*}{4} & \multirow{2}[0]{*}{6.77} & \multirow{2}[0]{*}{11} & $\m{\gamma}^* = \{0     ,0   ,  1    , 0   ,  0   ,  0   ,  0  ,   0  ,   1  ,   0\}$ \\
		&       &       &       &       & $\m{\pi}^* = \{0     ,0  ,   0  ,   1   ,  0   ,  0   ,  0   ,  0   ,  1  ,   0\}$ \\          \hline
		\multirow{2}[0]{*}{BSA-PBH} & \multirow{2}[0]{*}{-1.41$\times 10^{-2}$} & \multirow{2}[0]{*}{4} & \multirow{2}[0]{*}{2.68} & \multirow{2}[0]{*}{6} & $\m{\gamma}^* = \{0     ,0  ,   1   ,  0   ,  0   ,  0   ,  0  ,   0  ,   1  ,   0\}$ \\
		&       &       &       &       & $\m{\pi}^* = \{0     ,0   ,  1  ,   0  ,   0  ,   0   ,  0   ,  0  ,   1  ,   0\}$ \\
		\hline
	\end{tabular}%
	\label{tab:table1}%
\end{table*}%

\section{Numerical Experiments}~\label{sec:results}
We test the developed methods on a mass spring system~\cite{linfarjovTAC13admm,MihailoSiteMassSring} that consists of $N=10$ subsystems with $\mC=\mI$. All the simulations are performed using \iffalse MATLAB R2016b running on a 64-bit Windows 10 with 3.4GHz Intel Core i7-6700 CPU and 16 GB of RAM \else MATLAB R2017b running on a 64-bit Windows 10 with 2.5GHz Intel Core i7-6500U CPU and 8 GB of RAM\fi, where each optimization problem is solved using YALMIP \cite{lofberg2004yalmip} with MOSEK version 8.1~\cite{mosekAps}. Here, we impose a logistic constraint so that there are at least 2 activated sensors and 2 activated actuators. In this simulation, we consider three different scenarios that follow from the developed approaches in the previous sections:

\begin{itemize}
	\item The first scenario (MI-SDP) is carried out by solving problem \eqref{eq:Big-M-SaAProblem} via YALMIP's MI-SDP branch and bound~\cite{lofberg2004yalmip}. We choose $L_1 = 10^5$, $L_2 = 10^5$, $L_3 = 10^5$, $\epsilon_1 = 10^{-9}$, and $\epsilon_2 = 10^{-6}$ such that the left-hand side of \eqref{eq:Big-M-SaAProblem-2} is upper bounded by $-\epsilon_1\mI$ and $\m P \succeq \epsilon_2\mI$. Smaller values for $L_1$ and $L_2$ resulted in large computational time. 
	\item The second scenario (BSA-SDP) directly follows Algorithm \ref{alg:DaC} and solves \eqref{eq:OFS} in each iteration to check the feasibility of the given combination of SaA, while also computing the static output feedback gain matrix simultaneously from the solution of LMIs~\eqref{eq:OFS}.
	\item The third scenario (BSA-PBH) uses the modified version of Algorithm~\ref{alg:DaC}, as explained in section \ref{BSA-2}, along with the PBH tests. When the algorithm terminates, the obtained SaA solutions are tested to solve \eqref{eq:OFS}. The combination that is feasible for \eqref{eq:OFS} and has the least number of active SaA is then reported as the solution.
	
\end{itemize} 

\begin{figure}[t]
	\centering	\includegraphics[scale=0.531]{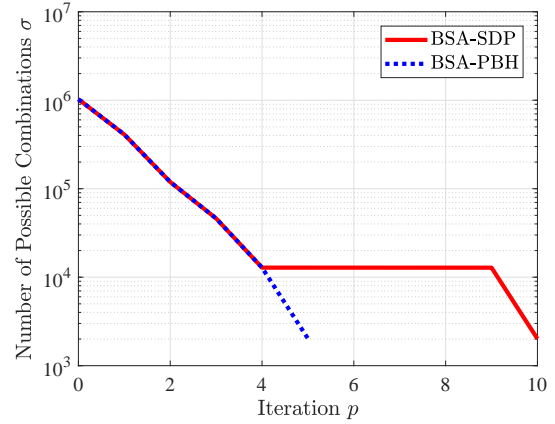}
		\caption{The reduction of the number of possible combinations. The algorithm terminates when $\sigma = 0$.}
		\label{fig1}
\end{figure}

The results of this numerical tests are presented in Table \ref{tab:table1}. All scenarios successfully return optimal solutions: 2 sensors and 2 actuators. Among these scenarios, the MI-SDP takes the longest time to compute an optimal solution. The BSA-PBH outperforms the other two scenarios in terms of computational time, while also taking fewer iterations compared to the BSA-SDP. This occurs because in BSA-SDP, problem \eqref{eq:OFS} is solved in each iteration, whereas BSA-PBH only checks the stabilizability and detectability of a given combination of SaA, a process that does not require much computations compared to solving SDPs. The reduction of the number of possible combinations of SaA between BSA-SDP and BSA-PBH is depicted in Figure \ref{fig1}. Note that the algorithm terminates when the candidate set is empty.
 
\section{Summary and Future Work}
Two general approaches to minimize the number of selected SaA for static output feedback stabilization are proposed. The first approach is based on solving a MI-SDP, while the second one uses a simple algorithm based on the binary search algorithm. The numerical tests on a mass spring system show that both approaches are able to give optimal solutions for the SaA selection problem.

Our future work will focus on investigating the scaling of the proposed methods into larger dynamic networks. Solving the MI-SDP of problem \eqref{eq:Big-M-SaAProblem} might consume large computational resources for larger systems. Also, a limitation of Algorithm~\ref{alg:DaC} is that it requires traversing the database of all possible SaA combinations. To that end, we plan to develop heuristics so that SaA selection problem can be applied for larger dynamic networks.


 
\bibliographystyle{IEEEtran}	\bibliography{bibliography}

\end{document}